\documentclass[10pt, onecolumn]{article}
\usepackage{multirow}
\usepackage{subfigure}
\usepackage{graphicx}
\usepackage{longtable}
\usepackage{courier}
\usepackage{float}
\usepackage{cite}
\usepackage{amsmath,amssymb}
\usepackage{bbold}
\usepackage{dsfont}
\usepackage[in]{fullpage}
\usepackage[english]{babel}
\usepackage{amsthm}
\usepackage{changes}

\newtheorem{theorem}{Theorem}
\newtheorem{corollary}{Corollary}[theorem]
\newtheorem{definition}{Definition}

\usepackage{wrapfig}

\usepackage[T1]{fontenc}
\usepackage{authblk}
\usepackage{array}
\usepackage{url}

\usepackage[pdfborder={0 0 0}]{hyperref}
\usepackage[margin=0.55in]{geometry}
\newcommand{\ou}{%
  \mathrel{%
    \vcenter{\offinterlineskip
      \ialign{##\cr$<$\cr\noalign{\kern-1.5pt}$>$\cr}
			 }%
  }%
}	
 \vspace{-8ex}
  \date{}

\begin{document}

 \title{Chernoff Information as a Privacy Constraint for Adversarial \\ Classification and Membership Advantage}
\author{%
Ay\c{s}e \"{U}nsal \\ \url{unsal@eurecom.fr}
}
\maketitle
 \pagestyle{plain}

  \begin{abstract}
This work investigates a privacy metric based on Chernoff information motivated by its importance in characterizing the optimal classifier's performance. Adversarial classification centers on minimizing the probability of error -either average error or correct detection error- when deciding between two classes in the binary setting. Classical hypothesis testing treats false alarm and mis-detection probabilities separately, resulting in asymmetric optimal error exponents. In this work, we instead characterize the relationship between $\varepsilon\textrm{-}$differential privacy, the optimal error exponent of one error probability conditioned on the other, and the optimal average error exponent.
To this end, we re-derive Chernoff differential privacy in connection with $\varepsilon\textrm{-}$differential
privacy using the Radon-Nikodym derivative and establish its relationship with Kullback-Leibler (KL) differential privacy to prove that Chernoff differential privacy is sandwiched between the two. We then present numerical evaluations demonstrating that Chernoff information outperforms the KL divergence as a function of the privacy parameter $\varepsilon$, particularly in capturing the impact of adversarial attacks under Laplace mechanisms.
Finally, we upper bound the adversary's advantage
in membership inference attacks based on Chernoff differential privacy and numerically
compare its performance with existing bounds derived using $(\varepsilon,\delta)\textrm{-}$differential privacy. Accordingly, we re-derive Chernoff differential privacy in connection with
$\varepsilon\textrm{-}$differential privacy using the Radon-Nikodym derivative, and prove its relation with
Kullback-Leibler (KL) differential privacy. Subsequently, we present
numerical evaluation results, which demonstrates that Chernoff information outperforms
Kullback-Leibler divergence as a function of the privacy parameter $\varepsilon$ and the impact of
the adversary's attack in Laplace mechanisms. Lastly, we introduce a new upper bound on
adversary's membership advantage in membership inference attacks using Chernoff DP and
numerically compare its performance with existing alternatives based on $(\varepsilon,\delta)\textrm{-}$differential
privacy in the literature.
\end{abstract}

\maketitle

\section{Introduction}\label{S:one}
The rapid and well-deserved success of machine learning (ML) applications-particularly over the past decade-has come with significant costs as a consequece of achieving high performance. Modern ML methods rely heavily on large-scale datasets to deliver accurate and efficient results, a dependence that increasingly jeopardizes the privacy and security of individuals who contribute-knowingly or unknowingly-to online data collections. As these applications become deeply embedded in everyday life across various sensitive domains, concerns over personal data privacy intensify, leaving data owners increasingly vulnerable to privacy breaches and has compelled robust privacy guarantees. 

Adversarial ML \cite{AdversarialML} studies privacy and security attacks and develop defense strategies to counter these attacks. Introducing adversarial examples to ML systems is a specific type of sophisticated and powerful attack, where additional -sometimes specially
crafted- or modified inputs are provided to the system with the intent of being misclassified
by the model as legitimate. To counter adversarial attacks that aim to manipulate data,
adversarial classification seeks to correctly detect such adversarial examples. In this setting,
the adversary's goal is to deceive a classifier that is originally designed to detect outliers.
Beyond the security implications of these misclassification attacks, user data privacy is also
subject to violations. This creates an interplay between the adversary's
 and the classifier's or defender's perspective in the study of adversarial
classification. 

Originally, classification theory covers the problem of optimally placing observations into different
categories, which are called classes. Each class is defined due to an optimal rule according
to some probabilistic description, which may or may not be subjected to some unknown
parameter(s). The main challenge here is to determine the optimal classifier's performance.
This performance is commonly characterized in connection with error probabilities in choosing
between one of the classes. In this paper, we employ the best average error exponent, namely the
Chernoff information/divergence, as a data privacy metric in classifying adversarial examples
targeting ML algorithms.

Another potential privacy risk is the membership inference attacks (MIAs), which involve repeatedly
querying a trained model to infer whether a specific data sample was included in its training
set, given access to a trained model and a target sample, the attacker probes the input space
and analyzes the model's outputs to determine the sample's membership in the training
data\cite{Shokri, Li13, Humphries23, MIA_survey}. MIAs are foundational privacy attacks targeting ML algorithms because it flags if there is some memorization of the training data or if the targeted model contains some information about
it. MIAs can be the backbones of other types of adversarial attacks, in other words protecting against MIAs leads to prevent against other attacks as well. One of the most powerful mitigation technique against it is the use of differential privacy (DP).
DP, originally introduced in \cite{DMNA06}, is the mathematical foundation of user data privacy. A randomized algorithm, also called a mechanism, guarantees the data to be analyzed without revealing personal information of any of the participants by employing DP. Essentially, a \textit{differentially private} mechanism ensures the level of the privacy of its individual participants and the output of the analysis to remain unaltered, even when \textit{any} user decides to leave or join the dataset with their personal information. 

In this paper, we introduce a privacy metric based on Chernoff information and analyze it through the Radon-Nikodym derivative framework. We compare this metric with Dwork's $\varepsilon$-DP and the Kullback-Leibler-based DP (KL-DP) proposed in \cite{CY16}, and analytically establish that Chernoff DP is sandwiched between $\varepsilon$-DP and KL-DP. Additionally, we derive a novel upper bound on the adversary's advantage in MIAs. We then numerically evaluate the proposed bound against existing theoretical guarantees in two distinct settings.
The main contributions of this work are summarized as follows:
\begin{enumerate}
\item We formulate Chernoff DP as a function of the Radon-Nikodym derivative and analyze its relationship with Dwork's $\varepsilon$-DP. We then characterize their analytical comparison in terms of the privacy budget and prove that $\varepsilon$-DP implies Chernoff DP.
\item We prove that Chernoff DP provides a stronger privacy guarantee than KL-DP.
\item We show that key structural properties-namely, the symmetry of Chernoff information and the composition property of differential privacy-are preserved under this formulation for both KL-DP and Chernoff DP.
\item We provide a numerical comparison of several divergence-based DP metrics.
\item We numerically evaluate the performance of $\varepsilon$-DP and Chernoff DP in terms of attack impact, global sensitivity of the query, and privacy budget.
\item We derive a new upper bound on the adversary's advantage in MIAs based on Chernoff DP and compare it numerically with existing bounds from the literature.
\end{enumerate}

\paragraph{Related Work}
The adaptation of well-known information-theoretic quantities as privacy metrics is not new. In addition to \cite{CY16}, which defines Shannon's mutual information as a privacy metric and proves that it coincides with $\varepsilon$-DP under certain conditions, other works have employed Rényi divergence \cite{Renyidp}, mutual information \cite{CY16, WYZ16, BK11, M12}, Kullback-Leibler divergence \cite{UO2021, UO22}, and min-entropy \cite{AA11, BK11} as privacy measures. An extensive tutorial on information-theoretic approaches to DP is provided in \cite{Unsal_CS}. Kullback-Leibler divergence and Chernoff information are particularly significant in classification problems, as they characterize the optimal error exponents for misdetection and average error probability, respectively. Kullback-Leibler divergence, along with other statistical divergence measures, has been widely used in classification settings \cite{JG1, nov24, Duchi16, Moreno, Hero2002}. To the best of our knowledge, however, Chernoff information has not been previously applied to adversarial classification or membership inference advantage. Recent works \cite{chernoff_fair, chernoff_fair2} employ Chernoff information based metric to study privacy–fairness trade-offs in Gaussian mechanisms.

MIAs, as fundamental privacy attacks against machine learning models, have also been extensively studied in a variety of settings \cite{Shokri, MIA_survey, Zari22, Li13}. Several works \cite{Yeom18, BayesSec, Erl20, Hall13, Humphries23} propose performance metrics to quantify attack success; however, these measures are typically derived under the $(\varepsilon, \delta)$-DP framework. In contrast, the present work adopts an information-theoretic perspective on DP.

\paragraph{Outline}
Section \ref{sec:preliminaries} starts off with properties and various definitions of DP and continues with basic preliminaries on hypothesis testing. Section \ref{sec:properties} introduces Chernoff DP and its properties. Our main result on the relation between Chernoff-DP, $\varepsilon\textrm{-}$DP and KL-DP is presented in Section \ref{sec:compare}. We re-visit the problem of adversarial classification in Section \ref{sec:adv_class} followed by Section \ref{sec:MIA} on various upper bounds on adversary's membership advantage. Lastly, we draw conclusions of our findings and briefly discuss the future work in Section \ref{sec:conc}.

\section{Preliminaries and Background \label{sec:preliminaries}}
This section presents the necessary preliminaries and establishes the theoretical background required for the remainder of the paper. In particular, we review some key concepts from measure theory, statistics, and information theory that will be used throughout our analysis.

In probability and measure theory, a function $\mu$ on a field $\mathcal{F}$ in a space $\Omega$ is called a measure given that the following conditions are met:
\begin{enumerate}
\item $\mu(A) \in [0,\infty]$ for $A \in \mathcal{F}$ ; 
\item $\mu(\emptyset)=0$;
\item if the sequence $\mathcal{F}\textrm{-}$sets $A_1, A_2,\cdots$ are a disjoint sequence of $\mathcal{F}\textrm{-}$sets where $\cup_{k=1}^{\infty} A_k \in \mathcal{F}$, then the following equality holds
\begin{equation}
\mu \left(\cup_{k=1}^{\infty} A_k\right)= \sum_{k=1}^{\infty}\mu(A_k).
\end{equation}
\end{enumerate} The pair $(\Omega, \mathcal{F})$ is called a \textit{measurable space} if $\mathcal{F}$ is a $\sigma\textrm{-}$field in the sample space $\Omega$ \cite{Prob_Measure}.
If a measure $P$ for $P(A)$ equals 0 whenever another measure $Q(A)$ equals 0, then $P$ is said to be dominated by another measure $Q$, denoted as $P<<Q$. For $P<<Q$, the Radon-Nikodym derivative of $P$ with respect to (w.r.t.) $Q$ is denoted by $\frac{dP}{dQ}$ \cite{Nikodym}.
\begin{definition}{[$(\varepsilon,\delta)\textrm{-}$closeness \cite{CY16}]} \label{def:Radon-Nikodym}
Probability distributions $P$ and $Q$ defined over the same measurable space $(\Omega, \mathcal{F})$ are called $(\varepsilon,\delta)\textrm{-}$close denoted by $P \overset{(\varepsilon,\delta)}{\approx}Q$
if the following couple of inequalities hold for any $A\in \mathcal{F}$.
\begin{align}
P(A) &\leq e^{\varepsilon} Q(A) +\delta \notag \\
Q(A) &\leq e^{\varepsilon} P(A) +\delta \notag
\end{align} 
\end{definition} When $\delta=0$, $(\varepsilon, 0)\textrm{-}$ closeness between $P$ and $Q$ can be represented by the Radon-Nikodym derivative as follows:
\begin{equation}\label{eq:RN_int}
e^{-\varepsilon}\leq \frac{dP}{dQ}(a)\leq e^{\varepsilon}, \; \forall a \in \Omega.
\end{equation} Equivalently, we have $\left|\log \frac{dP}{dQ}(a) \right| \leq \varepsilon$. In this paper, $\log(.)$ denotes the natural logarithm function.
\begin{definition}{[Neighboring datasets \cite{DR05}]} \label{eq:distance}
Any two datasets $x, \tilde{x}$ that differ only in one record are called neighbors. For two neighboring datasets, the equality $d(x, \tilde{x})=1$ holds, where $d(.,.)$ denotes the Hamming (or $l_1$) distance between two datasets. 
\end{definition} Definition \ref{eq:distance} anticipates symmetry among neighbors in terms of the size of the datasets. This could be further relaxed to include the datasets of different sizes, where neighborhood is due to addition or removal of a record. Both definitions ensure that the neighbors differ in a single row (in one user's data).
\begin{definition}{[$(\varepsilon,\delta)-$ differential privacy \cite{DR05}]} \label{def:dp_dwork}
A randomized algorithm $\mathcal{M}$ is $(\varepsilon, \delta)-$ differentially private if $\forall S \subseteq Range(\mathcal{M})$ and $\forall x, \tilde{x}$ that are neighbors within the domain of $\mathcal{M}$, the following inequality holds.
\begin{equation}\label{ineq:dp}
\Pr\left[\mathcal{M}(x) \in S\right] \leq \Pr\left[\mathcal{M}(\tilde{x}) \in S \right] e^{\varepsilon}+\delta.
\end{equation}
\end{definition} The randomized mechanism $\mathcal{M}$ can also be represented by the conditional distribution of the dataset $X^n=(X_1,X_2,\cdots,X^n)$ with the corresponding output $Y$ as $P_{Y|X^n}$. In this case, $P_{Y|X^n}$ satisfies $(\varepsilon, \delta)-$DP for all neighboring $x^n$ and $\tilde{x}^n$ if the following holds:
\begin{equation}\label{def:dp_closeness}
P_{Y|X^n} \overset{(\varepsilon, \delta)}{\approx} P_{Y|\tilde{X}^n}
\end{equation}  
Next, we remind the reader of the Kullback-Leibler DP (KL-DP) and mutual information differential privacy (MI-DP).
\begin{definition}{[KL-DP, \cite{CY16}]} \label{def:KL-DP}
A randomized mechanism $P_{Y|X}$ is said to guarantee $\varepsilon\textrm{-}$ KL-DP, if the following inequality holds for all its neighboring datasets $x$ and $\tilde{x}$, 
\begin{equation}
D(P_{Y|X^n}||P_{Y|\tilde{X}^n}) \leq \varepsilon.
\end{equation}
\end{definition}

\begin{definition}{[MI-DP \cite{CY16}]} \label{def:MI-DP}
$\varepsilon\textrm{-}$MI-DP holds for the randomized mechanism $P_{Y|X^n}$ if the following inequality is satisfied
\begin{equation}
\sup_{i, P_{X^n}} I(X_i;Y|X^{-i})\leq \varepsilon
\end{equation} where $X^{-i}$ denotes the sequence of $n$ random variables excluding $X_i$. 
\end{definition} The unit is expressed in nats rather than bits, as the natural logarithm is used throughout.
\subsection{Order of DP metrics}
The primary contribution of \cite{CY16} is a systematic comparison of $\varepsilon$-DP, KL-DP, and MI-DP, along with a characterization of their relative strength as privacy metrics. Formally, consider two privacy notions $a$-DP and $b$-DP with parameters $a,b>0$. We say that $a$-DP is \emph{stronger} than $b$-DP, and write
\begin{equation}
a\textrm{-}\mathrm{DP} \succeq b\textrm{-}\mathrm{DP}
\end{equation}
if for any positive $b'$, there exists a positive $a'$, such that 
\begin{equation}
a'\textrm{-}\mathrm{DP} \Longrightarrow b'\textrm{-}\mathrm{DP}.
\end{equation} In other words, $a\textrm{-}\mathrm{DP}$ is stronger than $b\textrm{-}\mathrm{DP}$ since $a\textrm{-}\mathrm{DP}$ implies $b\textrm{-}\mathrm{DP}$ for any non-negative and non-zero $b'$. Using this ordering, \cite[Theorem 1]{CY16} establishes the following chain of implications:
\begin{equation} \label{ineq:chain}
\varepsilon\textrm{-}\mathrm{DP} \succeq \mathrm{KL}\textrm{-}\mathrm{DP} \succeq \mathrm{MI}\textrm{-}\mathrm{DP} \succeq (\varepsilon,\delta\textrm{-})\mathrm{DP}
\end{equation} Here, the notation $\succeq$ indicates that the privacy guarantee on the left-hand side is stronger than that on the right-hand side, in the sense that the former implies the latter (up to an appropriate reparameterization).

\subsection{Hypothesis Testing \label{subsec:hypo_test}}

As one of the most widely used divergence measures, the Kullback-Leibler (KL) divergence \cite{Kullback, KL} between two probability measures $P$ and $Q$, both dominated by a common measure $\mu$, is defined as
\begin{equation} \label{eq:KL}
D(P||Q)= \int p \log \frac{p}{q} d\mu.
\end{equation}
Next, we define Chernoff information, also referred as Chernoff divergence, as follows.
\begin{definition}{[Chernoff Information \cite{Chernoff}]} \label{chernoff_def}
Chernoff information between any two probability measures $P$ and $Q$ with the dominating measure $\mu$ and parameter $\alpha$ is 
\begin{equation} \label{eq:chernoff_def}
C(P,Q)= \max_{\alpha \in (0,1)} - \log \int p^{\alpha}q^{1-\alpha} d\mu.
\end{equation}
\end{definition} (\ref{eq:chernoff_def}) can equivalently be expressed in terms of the logarithm of the $\alpha$-skewed Bhattacharyya coefficient $C_{\alpha}(P,Q)$, yielding 
\begin{equation}\label{eq:def_cher}
C(P,Q)= \max_{\alpha \in (0,1)} - \log C_{\alpha}(P,Q).
\end{equation} An important feature of Chernoff information-one that is not shared by the Kullback-Leibler divergence-is its symmetry, namely, $C(P,Q)=C(Q,P)$.

The significance of the divergence measures defined in (\ref{eq:KL}) and (\ref{eq:chernoff_def}) is established by Stein's lemma \cite{Chernoff, r14} in the contexts of classical and Bayesian binary hypothesis testing, respectively. Consider two competing hypotheses, $H_0$ and $H_1$, corresponding to probability distributions that are to be distinguished based on observed data. The probability of false alarm and the probability of misdetection are defined as
\begin{equation} P_{fa}=\Pr[\textrm{Choose}\; H_1|H_0\; \textrm{correct}], \quad P_{miss}=\Pr[\textrm{Choose}\; H_0|H_1\; \textrm{correct}].
\end{equation} When a priori probabilities are assigned to the hypotheses, the overall (Bayesian) average error probability is given by
\begin{equation}
P_e= \pi P_{fa} + (1-\pi) P_{miss}
\end{equation} for some $\pi \in (0,1)$. The optimal classifier obeys the following asymptotics for an $M-$dimensional random vector of observations;
\begin{align}
\lim_{M\rightarrow \infty}\frac{P_{fa}}{M} & = -D(Q||P),\;\;\textrm{for fixed } P_{miss}, \label{eq:P_fa} \\
\lim_{M\rightarrow \infty}\frac{P_{miss}}{M} & = -D(P||Q),\;\;\textrm{for fixed }P_{fa},  \label{eq:P_miss}\\
\lim_{M\rightarrow \infty}\frac{P_{e}}{M} & = -C(P,Q). \label{eq:P_av}
\end{align} In addition to the relations between error exponents and divergence functions, Chernoff showed in \cite{Chernoff} that $C(P,Q)$ can be used to obtain Kullback-Leibler divergence as follows:
\begin{equation}
\left[\frac{d C_{\alpha}(P,Q)}{d\alpha}\right]_{\alpha=0}=D(Q||P).  \label{eq:cher_kl1}
\end{equation} Similarly, the tangent at $\alpha=1$ gives $-D(P||Q)$. 
In this paper, we re-derive Chernoff information as a privacy metric and demonstrate how it relates and compares to $\varepsilon\textrm{-}$DP and KL-DP.
\section{Chernoff DP, KL-DP and Composability \label{sec:properties}}

We begin by restating the notion of Chernoff differential privacy introduced in \cite{UO22}.
\begin{definition}{[Chernoff-DP]}\label{def:C-DP}
A randomized mechanism $P_{Y|X}$ with input $X$ and output $Y$ guarantees $\varepsilon$-Chernoff differential privacy if, for every pair of neighboring datasets $x$ and $\tilde{x}$,
\begin{equation}
C(P_{Y|X^n},P_{Y|\tilde{X}^n})\leq \varepsilon,
\end{equation} for some $\varepsilon > 0$.
\end{definition}
The main objective of this paper is to determine how Chernoff-DP fits into the chain of inequalities in (\ref{ineq:chain}) established in \cite{CY16}, thereby clarifying its relative position among divergence-based privacy metrics.
In \cite{UO22}, we also considered a relaxed variant of this definition based on the logarithm of the $\alpha$-skewed Bhattacharyya affinity coefficient. Since the maximization over $\alpha$ is upper bounded by $\varepsilon$, the corresponding Chernoff information is likewise bounded above by $\varepsilon$ for all $\alpha$.

A natural question concerns the choice of divergence measure. In the specific case of adversarial classification, as in standard binary hypothesis testing, the Chernoff information characterizes the optimal exponential decay rate of the average error probability achieved by the Bayes classifier \cite{Chernoff}. In other words, it determines the best achievable error exponent under optimal decision rules. By contrast, the closest related Renyi divergence-based differential privacy measure \cite{Renyidp} does not directly quantify the probability of correctly detecting data manipulation in the presence of an adversarial attack, and therefore provides no explicit characterization of detection performance.

One of the fundamental properties of differential privacy is \emph{composability}, which prevents unbounded privacy leakage under repeated analyses. Intuitively, composition quantifies how privacy degrades when multiple mechanisms are applied to the same dataset. In particular, the \emph{sequential composition} property states that if several mechanisms, each corresponding to a different query, individually satisfy DP, then their joint release also satisfies DP. Moreover, the resulting privacy budget scales proportionally with the number of queries. Sequential composition \cite{DKMMN, KOV15} therefore provides an upper bound on the cumulative privacy loss incurred by releasing a sequence of outputs from DP mechanisms applied to the same data. For completeness, we recall this result in the following theorem.
\begin{theorem} For a collection of $m$ randomized algorithms $\mathcal{M}(x) =(\mathcal{M}_1(x), \cdots, \mathcal{M}_m(x))$ defined over $\mathcal{X}^n \rightarrow \mathcal{Y}^m$ and is run (over the same input data) independently, composability property of DP ensures that $\mathcal{M}$ satisfies $\sum_j^m \varepsilon_j\textrm{-}$DP.
\end{theorem}
The proof follows directly from the multiplicative property of independent events when forming a product distribution over a tuple. An analogous argument establishes the corresponding lower bound by considering the negative exponent. 
The additivity of the Kullback-Leibler divergence immediately yields a sequential composition result when KL divergence is used to define the privacy constraint, as in Definition \ref{def:KL-DP}.
\begin{corollary} \label{cor:KLComp}
Let $\{Y_j\}$ for $j=1,\cdots,m$ be a collection of conditionally independent query outputs given the dataset, where each mechanism $P_{Y_j|X^n}$ satisfies $\varepsilon_j$-KL-DP, then the joint mechanism also satisfies KL-DP with a privacy budget of $\sum_j^m \varepsilon_j = \varepsilon$.
\end{corollary}
The proof follows directly from the additivity of KL divergence under conditional independence.
\begin{proof}
The collection of $m$ conditionally independent mechanisms that satisfy $\varepsilon_j\textrm{-}$KL-DP is given by
\begin{align}
D(P_{Y^m|X^n}||P_{Y^m|\tilde{X}^n})
&= \sum_{j=1}^m D(P_{Y_j|X^n}||P_{Y_j|\tilde{X}^n})\label{proof_1} \\
& \leq \sum_j^m \varepsilon_j \label{proof_2}
\end{align} which is equal to $\varepsilon$.
(\ref{proof_1}) follows from the fact that the set of statistically independent query outputs given the dataset represent the mechanism $P_{Y^m|X^n}$ which is the product of the individual mechanisms for each $j$. Substituting the Definition \ref{def:KL-DP} into (\ref{proof_1}) yields (\ref{proof_2}). The rest is a result of the following property of the logarithmic function in the definition of Kullback-Leibler divergence; $\log (a \times b)=\log a + \log b$. \end{proof}
Similarly, MI-DP as defined in Definition \ref{def:MI-DP} is shown in \cite{CY16} to satisfy a composition theorem by invoking the chain rule of mutual information. In general, Chernoff information as defined in Definition \ref{chernoff_def} is not additive, due to the maximization over its parameter $\alpha$. This contrasts with the Kullback-Leibler divergence and the Bhattacharyya distance (the latter corresponding to Chernoff information with $\alpha = 0.5$), both of which are additive under independence.
We now state the sequential composition property for Chernoff-DP.
\begin{corollary} \label{cor:chernoff_comp}
Suppose that, for each $j=1,\dots,m$, the randomized mechanism $P_{Y_j|X^n}$ satisfies $\varepsilon_j$-Chernoff DP. If the mechanisms are conditionally independent given the dataset, then their joint release satisfies Chernoff-DP with overall privacy budget $\sum_{j=1}^m \varepsilon_j$.
\end{corollary}
\begin{proof}
By definition, Chernoff-DP is expanded out as follows:
\begin{equation}
\max_{\alpha \in (0,1)} -\log C_{\alpha}(P_{Y|X^n},P_{Y|\tilde{X}^n}) \leq \varepsilon_k \label{1st_line}
\end{equation} The maximization on the l.h.s. of (\ref{1st_line}) can be removed without having any effect on the upper bound since it upper bounds the function for any value of $\alpha$ in its range.
\begin{equation}
-\log C_{\alpha}(P_{Y|X^n},P_{Y|\tilde{X}^n}) \leq \varepsilon_k
\end{equation}  Substituting conditionally independent randomized mechanisms $P_{Y^m|X^n}$ and $P_{Y^m|\tilde{X}^n}$ after removal of the maximization function, we have for the l.h.s.
\begin{align}
-\log \int P_{Y^m|X^n}^{\alpha} P_{Y^m|\tilde{X}^n}^{1-\alpha} d\mu&= -\log \prod_{j=1}^m \int  P_{Y_j|X^n}^{\alpha} P_{Y_j|\tilde{X}^n}^{1-\alpha} d\mu \notag \\
&\leq \sum_{j=1}^m \varepsilon_j
\end{align} It is worth noting that, unlike $\varepsilon\textrm{-}$DP that is confined in the interval $[e^{-\varepsilon}, e^{\varepsilon}]$, Chernoff-DP is lower bounded by 0.
Contrary to Chernoff-DP, Chernoff information defined by (\ref{eq:chernoff_def}) is not additive. \end{proof}
\section{Analytical Comparison and Ordering of Divergence-Based DP Metrics \label{sec:compare}}
In this section, we re-derive Chernoff-DP via $\varepsilon\textrm{-}$DP.
As shown in \cite{CY16}, KL-DP (Definition \ref{def:KL-DP}) can be positioned between $\varepsilon$-DP and $(\varepsilon,\delta)$-DP by expressing the Kullback-Leibler divergence as a function of the Radon-Nikodym derivative $\frac{dP}{dQ}$. Specifically, for two probability measures $P$ and $Q$, we have
\begin{align}
D(P||Q)&=\int dP(a) \log \frac{dP}{dQ}(a) \label{eq:first_line}\\
&= \int \left[\frac{dP}{dQ}(a)\log \frac{dP}{dQ}(a) dQ(a)\right] \label{eq:2ndline}\\
&= \mathbb{E}\left[Z \log Z \right]  \label{eq:last_line}
\end{align} where (\ref{eq:2ndline}) follows by rewriting the integral with respect to $Q$. Letting $X \sim Q$, we define the Radon-Nikodym derivative as the random variable 
\begin{equation}
Z=\frac{dP}{dQ}(X).
\end{equation} Under $\varepsilon\textrm{-}$DP, the likelihood ratio is bounded as $Z \in [e^{-\varepsilon}, e^{\varepsilon}]$ and satisfies the normalization constraint $\mathbb{E}[Z] = 1$. To characterize the extremal KL divergence under these constraints, $Z$ is concentrated at the boundary points $e^{-\varepsilon}$ and $e^{\varepsilon}$. In particular, $Z$ takes the value $e^{\varepsilon}$ with probability $p=\frac{1-e^{-\varepsilon}}{e^{\varepsilon}-e^{-\varepsilon}}$ and the value $e^{-\varepsilon}$ with probability $1-p$, ensuring that $\mathbb{E}[Z] = 1$. Substituting this two-point distribution into (\ref{eq:2ndline}) yields the KL divergence expressed explicitly as a function of the Radon-Nikodym derivative bounds for (\ref{eq:last_line})
\begin{equation}\label{eq:KL_cuff}
D(P||Q) \leq \varepsilon \frac{(e^{\varepsilon}-1)(1-e^{-\varepsilon})}{e^{\varepsilon}-e^{-\varepsilon}}
\end{equation} 
This expression characterizes the maximum KL divergence compatible with $\varepsilon$-DP. An immediate and notable consequence of this formulation is the induced symmetry between $D(P||Q)$ and $D(Q||P)$, when the divergence is interpreted as a privacy metric under the likelihood-ratio bounds imposed by $\varepsilon$-DP.

We now follow an analogous approach for Chernoff information. Specifically, we substitute the Radon-Nikodym derivative into the definition of Chernoff information in order to express it explicitly in terms of the $\varepsilon$-DP constraint.

\subsection{Characterizing Chernoff-DP through \texorpdfstring{$\varepsilon\textrm{-}$}{epsilon}DP}
The $\alpha$-skewed Bhattacharyya affinity coefficient, defined as 
\begin{equation}
C_{\alpha}(P,Q)=\int p^{\alpha} q^{1-\alpha} d\mu
\end{equation} can be rewritten in terms of $Q$ by a simple change of variables. Specifically, when $P \ll Q$, we have
\begin{equation}\label{def:Qbased}
\int p^{\alpha}q^{1-\alpha} d\mu=\int \left(p/q\right)^{\alpha}dQ
\end{equation} where $p$ and $q$ are the Radon-Nikodym derivatives of $P$ and $Q$ with respect to $\mu$.
Equivalently, $C_{\alpha}(P,Q)$ can be expressed in terms of $Q$, yielding an alternative form of the $\alpha$-skewed Bhattacharyya affinity coefficient. This representation leads to the same Chernoff information as defined in (\ref{eq:chernoff_def}).
\begin{equation}\label{def:Pbased}
\int p^{\alpha}q^{1-\alpha} d\mu=\int \left(q/p\right)^{1-\alpha}dP
\end{equation}
It is straightforward to rewrite (\ref{def:Qbased}) as
\begin{align}\label{Ch_Q_based}
C_{\alpha}(P,Q)&= \int \left(\frac{dP}{dQ}(a)\right)^{\alpha} dQ(a) \\
&=\mathbb{E}\left[Z^{\alpha}\right] \label{eq:Q_based_2ndline}
\end{align} where in (\ref{eq:Q_based_2ndline}), we define the likelihood-ratio random variable
\begin{equation}
Z=\frac{dP}{dQ}(X), \quad X \sim Q.
\end{equation} Plugging $C_{\alpha}(P,Q)$ into (\ref{eq:def_cher}), we finally get $C(P,Q)$
\begin{equation} \label{eq:optim_fail_eq}
C(P,Q)=\max_{\alpha} \left\{\log(\alpha+1)-\log\left(e^{\varepsilon(\alpha+1)}-e^{-\varepsilon(\alpha+1)}\right)\right\}.
\end{equation} 

Similarly, if the integral representation of $C_{\alpha}(P,Q)$ is written with respect to $P$ instead of $Q$, substituting the Radon-Nikodym derivative yields $C_{\alpha}(P,Q)=\mathbb{E}_P\left[1/Z^{1-\alpha}\right]$. Thus, the $\alpha$-skewed Bhattacharyya affinity coefficient admits equivalent representations under either measure.
Using this alternative representation of the $\alpha$-skewed Bhattacharyya affinity coefficient, we obtain
\begin{align}\label{Ch_P_based}
C_{\alpha}(P,Q)&= \int \left(\frac{dQ}{dP}(a)\right)^{1-\alpha} dP(a)\\
&= \mathbb{E}\left[Z^{\alpha-1}\right]
\end{align} where $X \sim P$ and the likelihood-ratio random variable is defined as $Z=\frac{dP}{dQ}(X)$.
Note that the representations in (\ref{Ch_Q_based}) and (\ref{Ch_P_based}) are equivalent; they simply correspond to expressing the integral with respect to different dominating measures.
The following theorem establishes the ordering of Chernoff-DP relative to KL-DP and $\varepsilon\textrm{-}$DP by expressing each notion as a function of the Radon-Nikodym derivative $dP/dQ$.
\begin{theorem}[Main result I]\label{main_theorem}
The following chain of implications holds
\begin{equation} \label{eq:chain}
\varepsilon\textrm{-}\mathrm{DP}  \succeq \textrm{Chernoff-}\mathrm{DP}\succeq \textrm{KL-DP}.
\end{equation}
\end{theorem} Theorem \ref{main_theorem} shows that $\varepsilon$-DP provides the strongest guarantee among the three, while KL-DP is the weakest, with Chernoff-DP lying strictly between them in terms of privacy strength.
\begin{proof}
In order to prove Theorem \ref{main_theorem}, we need to show that relations (a) and (b) hold:
\begin{equation}
\varepsilon\textrm{-}\mathrm{DP}  \overset{(a)}{\succeq} \textrm{Chernoff-}\mathrm{DP}\overset{(b)}{\succeq}\textrm{KL-DP}
\end{equation}
To obtain Chernoff DP via $\varepsilon\textrm{-}$DP of Definitions \ref{def:Radon-Nikodym} and \ref{def:dp_closeness} as the proof of (a), we need to re-write Chernoff DP as a function of $Z$.
With the expectation in (\ref{eq:Q_based_2ndline}) over its range $[e^{-\varepsilon},e^{\varepsilon}]$, $C_{\alpha}(P,Q)$ becomes
\begin{equation}
C_{\alpha}(P,Q)=\frac{1}{\alpha+1} \left(e^{\varepsilon(\alpha+1)}-e^{-\varepsilon(\alpha+1)}\right)\label{eq:Calpha}
\end{equation}
Substituting (\ref{eq:Calpha}) into Chernoff DP, for the first expansion as a function of (\ref{eq:Q_based_2ndline}), we get (\ref{eq:optim_fail_eq})
or equivalently,
\begin{equation}
C(P,Q)=\max_{\alpha}\left\{\log(\alpha+1)+\varepsilon(\alpha+1)-\log\left(e^{2\varepsilon(\alpha+1)}-1\right)\right\}.\;\;\; \label{eq:optim_fail}
\end{equation}
Jensen's equality states that for any convex function $g$ and a random variable $X$, the following inequality holds 
\begin{equation}
\mathbb{E}\left[g(X)\right] \geq g\left(\mathbb{E}[X]\right).
\end{equation} Note that, if $g(X)$ is concave, then $-g(X)$ is convex. Accordingly, $-\log \mathbb{E}\left[Z^{\alpha}\right]$ of $C(P,Q)$ is upper bounded by $\mathbb{E}\left[-\log Z^{\alpha}\right]$ as follows
\begin{align}
\mathbb{E}\left[-\log Z^{\alpha}\right]&=-\alpha \varepsilon \frac{1-e^{-\alpha \varepsilon}}{e^{\alpha \varepsilon}-e^{-\alpha \varepsilon}}+ \alpha \varepsilon \frac{e^{\alpha \varepsilon}-1}{e^{\alpha \varepsilon}-e^{-\alpha \varepsilon}}\\
&=\alpha \varepsilon \frac{e^{\alpha \varepsilon}+e^{-\alpha \varepsilon}-2}{e^{\alpha \varepsilon}-e^{-\alpha \varepsilon}}.\label{eq:finalform}
\end{align} We could rewrite (\ref{eq:finalform}) in relation to $C(P,Q)$ as
\begin{equation}
C(P,Q)\leq \alpha \varepsilon \frac{(e^{\alpha \varepsilon}-1)(1-e^{-\alpha \varepsilon})}{e^{\alpha \varepsilon}-e^{-\alpha \varepsilon}}. \label{eq:ch_proof}
\end{equation} This proves the left hand side (l.h.s.) of Chernoff-DP in the chain of inequalities (\ref{eq:chain}). As for proving (b), the right hand side (r.h.s.) of the chain of inequalities on the relation between Chernoff DP and KL-DP, we use the direct relation shown in (\ref{eq:cher_kl1}) \cite{Chernoff}. The reader should note that, originally KL divergence is asymmetric unlike KL-DP of (\ref{def:KL-DP}). Verifying either of the tangent at $\alpha=0$ in (\ref{eq:cher_kl1}) and at $\alpha=1$ will be sufficient to prove $\textrm{Chernoff-}\mathrm{DP}\succeq \textrm{KL-DP}$. Taking the derivative of the r.h.s. of (\ref{eq:ch_proof}) yields 0 for $\alpha=0$ and for $\alpha=1$, we get
\begin{equation}
\left[\frac{d \alpha \varepsilon \frac{(e^{\alpha \varepsilon}-1)(1-e^{-\alpha \varepsilon})}{e^{\alpha \varepsilon}-e^{-\alpha \varepsilon}} }{d\alpha}\right]_{\alpha=1}=\varepsilon \frac{e^{\varepsilon}-1}{e^{\varepsilon}+1} \label{eq:tangent1}
\end{equation} 
Simplifying (\ref{eq:KL_cuff}), we also have 
\begin{align}
D(P||Q) &\leq \varepsilon \frac{(e^{\varepsilon}-1)(1-e^{-\varepsilon})}{e^{\varepsilon}-e^{-\varepsilon}} \\
&=\varepsilon \frac{e^{2\varepsilon}-2e^{\varepsilon}+1}{e^{2\varepsilon}-1} \\
&=\varepsilon \frac{e^{\varepsilon}-1}{e^{\varepsilon}+1}\label{eq:KL_final}
\end{align} (\ref{eq:KL_final}) is equivalent to (\ref{eq:tangent1}). This completes the proof of (b) and Theorem \ref{main_theorem}.
\end{proof} 

\subsection{Sub-optimal solution by upper bounding the logarithmic function} 
Notably, optimization of \eqref{eq:optim_fail_eq}, which is obtained by substituting the Radon-Nikodym derivative into the Chernoff-DP formulation via (\ref{eq:Q_based_2ndline}), does not yield a closed-form expression characterizing the relationship between $\alpha$ and $\varepsilon$.
For the sake of completeness, in order to be able to define such relationship, we upper bound the Chernoff-DP in step (\ref{eq:optim_fail}), since the logarithmic function is monotonically increasing and the following inequality holds $\log(x+1) \leq x$ for any positive $x$. Hence, upper bounding the first term of (\ref{eq:optim_fail}) yields
\begin{equation} \label{eq:upper}
C_{ub}(P,Q)=\max_{\alpha}\{\alpha+\varepsilon(\alpha+1)-\log\left(e^{2\varepsilon(\alpha+1)}-1\right)\}
\end{equation} where $C(P,Q)\leq C_{ub}(P,Q)$. To find the sub-optimal $\alpha$, we differentiate $C_{ub}(P,Q)$ and solve for $\alpha$ in terms of $\varepsilon$ by setting the derivative equal to zero as follows. 
\begin{equation}
\frac{dC_{ub}(P,Q)}{d\alpha}=1+\varepsilon-\frac{2\varepsilon e^{2\varepsilon(\alpha+1)}}{e^{2\varepsilon(\alpha+1)}-1}
\end{equation} $\alpha$ that maximizes the upper bound on Chernoff-DP as a function of the privacy parameter is given by $\alpha_{ub,I}^{*}=\frac{1}{2\varepsilon}\log\frac{1+\varepsilon}{1-\varepsilon}-1$. Given that $\alpha$ is defined over $(0,1]$, $\alpha_{ub,I}^{*}$ also has to obey this range. On the l.h.s. of $0<\alpha_{ub,I}^*\leq1$, we have $\varepsilon >0$, whereas the r.h.s. yields
\begin{equation}\color[rgb]{0,0,0}\color[rgb]{0,0,0}
\frac{1}{2\varepsilon}\log\left(\frac{1+\varepsilon}{1-\varepsilon}\right)\leq 2. \label{eq:deriv1} 
\end{equation} Lower bounding the l.h.s. of (\ref{eq:deriv1}) by $\log x \geq 1-\frac{1}{x}$ for $x \in \mathbb{R}_{>0}$ yields $\varepsilon \geq -.5$ which obeys the range we obtain from $\alpha_{ub,I}^*>0$ and does not conflict with Definition \ref{def:C-DP}.

For the second expansion (\ref{Ch_P_based}), where the expectation of the Radon-Nikodym derivative is based on $P$, we achieve
\begin{equation}
C(P,Q)=\max_{\alpha} \left\{\log(\alpha)-\log \left(e^{\varepsilon \alpha}-e^{-\varepsilon \alpha} \right)\right\} 
\label{eq:optim_fail2}
\end{equation} for $C_{\alpha}(P,Q)=\frac{1}{\alpha} \left(e^{\varepsilon \alpha}-e^{-\varepsilon \alpha}\right)$. As in the first expansion with no explicit solution of maximization based on $\alpha$, via upper bounding the first natural logarithm for any positive real number via $\log x \leq x-1$, we obtain 
\begin{equation}
C_{\alpha}(P,Q)\leq \alpha-1+\varepsilon \alpha-\log\left(e^{2\varepsilon \alpha}-1\right).
\end{equation} with the sub-optimal solution for $\alpha$ given by $\alpha_{ub,II}^*=\frac{1}{2\varepsilon}\log \frac{1+\varepsilon}{1-\varepsilon}$. In order for $\alpha_{ub,II}^*$ to be valid, it must be confined in $(0,1]$. Accordingly, $\alpha^*_{ub,II}>0$ is guaranteed by a positive privacy parameter in accordance with Definition \ref{def:C-DP}. As for the upper bound on $\alpha_{ub,II}^*$, we obtain
\begin{equation}
\frac{1}{2\varepsilon}\log \frac{1+\varepsilon}{1-\varepsilon} \leq 1 \label{ineq:alpha}
\end{equation} which again yields $\varepsilon >0$. In (\ref{ineq:alpha}), the l.h.s. is lower bounded via the following property of natural logarithm function $\log x \geq 1-\frac{1}{x}$ where $x \in \mathbb{R}_{>0}$.
The symmetry in Chernoff information of Definition \ref{eq:chernoff_def} is preserved in Chernoff-DP hence, substitutions of both values of $\alpha$ into the upper bounds on two different expressions (\ref{eq:optim_fail}) and (\ref{eq:optim_fail2}) result in 
\begin{align}
C(P,Q) &\leq C^*_{ub}(P,Q) \notag \\
&=\left(\frac{1}{2\varepsilon}+\frac{1}{2}\right)\log \frac{1+\varepsilon}{1-\varepsilon}-1+\log \left(\frac{2\varepsilon}{1-\varepsilon}\right),
\end{align} where we denote the optimal upper bound on Chernoff-DP by $C^*_{ub}(P,Q)$. It is worth noting that, the symmetry property of Chernoff information is carried on the optimal upper bound on the DP metric due to the identical bounding step applied on both expansions. Theorem \ref{main_theorem} can be interpreted as necessary condition for $\varepsilon-$DP to imply Chernoff-DP which is dependent on the relation between the parameter $\alpha$ and the privacy parameter $\varepsilon$ to hold. Here the effect of different expansions are emphasized on the optimal value of parameter $\alpha$. 

\subsection{Numerical Evaluation} Figures \ref{fig:cher_bounds_high}-\ref{fig:cher_bounds_low} depict the numerical comparison of various divergence based DP metrics. Accordingly, Chernoff-Cuff refers to (\ref{eq:optim_fail_eq}). $D(P||Q)$ of (\ref{eq:KL_cuff}) is plotted via its product with $(1-\alpha)$ and is referred as KL-Cuff in the legend of Figures \ref{fig:cher_bounds_high} and \ref{fig:cher_bounds_low}. Additionally, $D(P||Q)$ and $-D(Q||P)$ correspond to respective derivatives of Chernoff-Cuff for $\alpha=0$ and $\alpha=1$. The corresponding curves are plotted via the products $\alpha \cdot D(Q||P)$ and $(1-\alpha)\cdot D(P||Q)$. It is clear from both plots that for smaller values of $\alpha$, Chernoff-Cuff of \eqref{eq:optim_fail_eq} is the tightest one among the alternatives. 
 \begin{figure}[ht]
	\centering
			\includegraphics[width=0.75\textwidth]{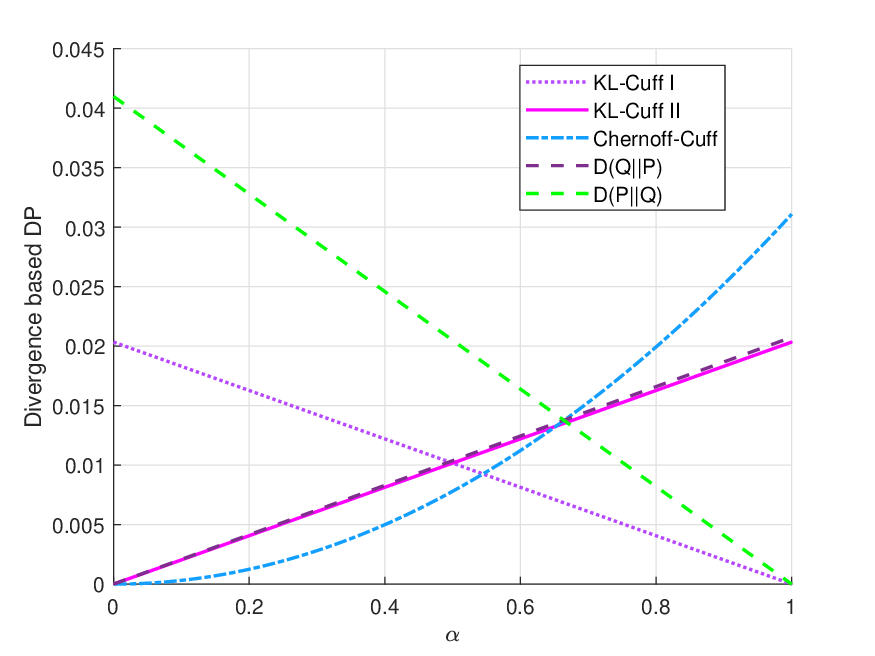}
			\caption{$\varepsilon=.25$.}
			\label{fig:cher_bounds_high}
\end{figure}
\begin{figure}
			\centering
			\includegraphics[width=0.75\textwidth]{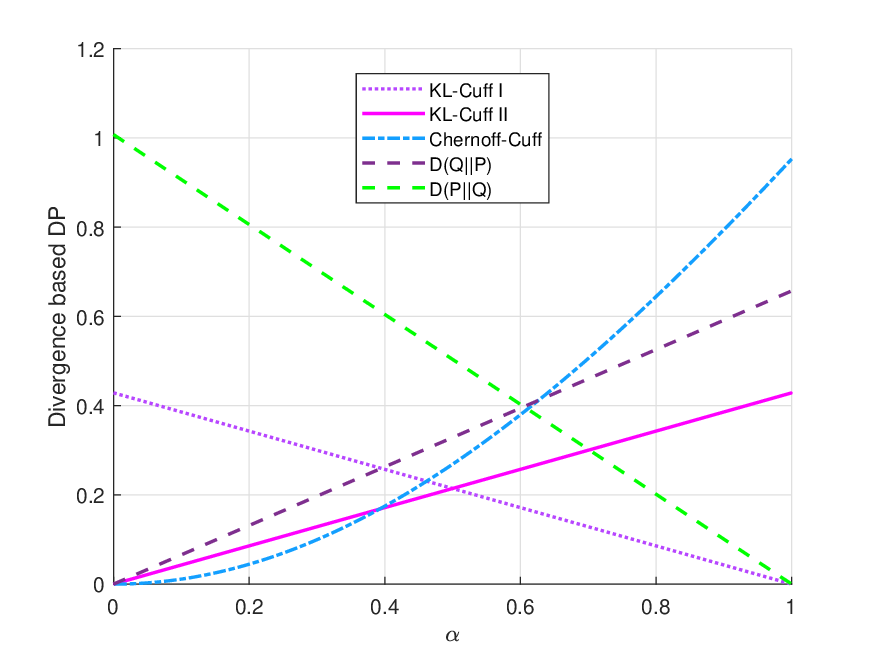}
			\caption{$\varepsilon=1.5$.}
			\label{fig:cher_bounds_low}
\end{figure}

\section{Adversarial Classification with DP \label{sec:adv_class}}

As discussed in Section \ref{subsec:hypo_test}, the motivation for considering privacy metrics based on Chernoff and KL divergences stems from their characterization of optimal error exponents in binary hypothesis testing. Consider an adversarial setting, in which a powerful attacker exploits the presence of a privacy mechanism to avoid detection. In particular, the adversary is assumed to be aware of the deployed privacy protection and aims to tune the magnitude of the attack according to the privacy budget of the mechanism in order to minimize the probability of being detected \cite{UO2021}. Here, the attack refers to a modification of the dataset, such as inserting, deleting or replacing a record. To provide a numerical comparison between KL-DP and Chernoff-DP in the context of adversarial classification under the Laplace mechanism, we briefly recall it below.
\begin{definition}
The Laplace mechanism \cite{DMNA06} for a query function $f: D \rightarrow \mathbb{R}^k$ is defined by adding independent Laplace noise components $N_i \sim \mathrm{Lap}(b=s/\varepsilon)$, $i=1,\dots,k$, to the query output:
\begin{equation}
\mathcal{Y}(x, f(.), \varepsilon)= f(x)+(N_1, \cdots, N_k)
\end{equation} where $s$ denotes the global sensitivity of $f$, 
\begin{equation}\|f(x)-f(\tilde{x})\|_1 \leq s,
\end{equation} for all neighboring datasets $x$ and $\tilde{x}$.
\end{definition}

Let the dataset be denoted by $X^n = {X_1,\dots,X_n}$. The released output under the Laplace mechanism is
\begin{equation}
\mathcal{Y}(x,f(.),\varepsilon)=Y=f(X^n)+N, \quad N \sim \mathrm{Lap}(0,b)
\end{equation} where $b = s/\varepsilon$.
Suppose an adversary alters the dataset by inserting or deleting a single record, denoted $X_a$. The modified output becomes
\begin{equation}
Y_a=f(X^n+X_a)+N.
\end{equation}
This setting can be formulated as a binary hypothesis testing problem, where the defender seeks to determine whether the dataset has been altered. The two hypotheses correspond to the presence or absence of the adversarial modification. Framing the problem in this way enables the defender to determine a detection threshold while simultaneously allowing the adversary to adjust the attack magnitude to evade detection \cite{UO2021}.

When $f(\cdot)$ is linear, even if the adversary only observes the published output, the likelihood ratio test can be constructed explicitly. The performance of this test is quantified by error probabilities that are directly characterized by either the KL divergence (in the classical setting) or the Chernoff information (in the Bayesian setting). The key distinction between the two formulations lies in the role of prior probabilities $\pi$ and $1-\pi$ assigned to the opposing hypotheses. However, as the number $M$ of i.i.d. observations grows, the influence of the priors vanishes asymptotically \cite{r14}, and the error exponent becomes independent of $\pi$. 

In this context, neighboring datasets correspond to the dataset before and after the adversarial modification. The associated probability distributions are therefore Laplace distributions with shifted means. Specifically, the classical test compares $N\sim Lap(\mu_0, b=s/\varepsilon)$, against $Lap(\mu_1, \theta b)$ for $\theta>1$ in the Bayesian formulation. The mean shift $\Delta\mu = \mu_1 - \mu_0$ results from the insertion or deletion of $X_a$, with $\mu_0 = 0$.
The Kullback-Leibler divergence between two Laplace distributions is derived in \cite{UO2021} as
\begin{equation}
D(P||Q)=\log \theta -1 +\frac{|\Delta \mu|}{\theta b}+\frac{1}{\theta}e^{-|\Delta \mu|/b}. \label{eq:KL-comp}
\end{equation} Similarly, the Chernoff information between two Laplace distributions \cite{JS03}, we get the following expression adapted to the described problem
\begin{equation}
C(P,Q)=\frac{|\Delta\mu|}{\theta b}-\log\left(1+\frac{|\Delta\mu|}{\theta b}\right).\label{eq:C-comp}
\end{equation}
\subsection{Numerical Evaluation} In Figure \ref{fig:compare_div}, we present a numerical comparison of (\ref{eq:KL-comp}) and (\ref{eq:C-comp}) for various values of $\Delta \mu$, highlighting their behavior as a function of the query sensitivity and the privacy budget.
The parameter $\theta$ represents the change in variance after the attack, while $\Delta\mu$ denotes the shift in the mean caused by the addition (or removal, if negative) of $X_a$, expressed as a multiple of the sensitivity $s$. As $\Delta\mu$ increases relative to the sensitivity, the divergence progressively approaches the upper bound. In the extreme case where $\Delta\mu = 3s$ and $\theta > 1$, both the Kullback-Leibler and Chernoff divergences exceed the upper bound $\varepsilon$. Across all three scenarios, the Chernoff divergence provides a significantly tighter bound than the Kullback-Leibler divergence as a function of the privacy budget.
\begin{figure}[ht]
 \centering
  \includegraphics[width=0.75\linewidth]{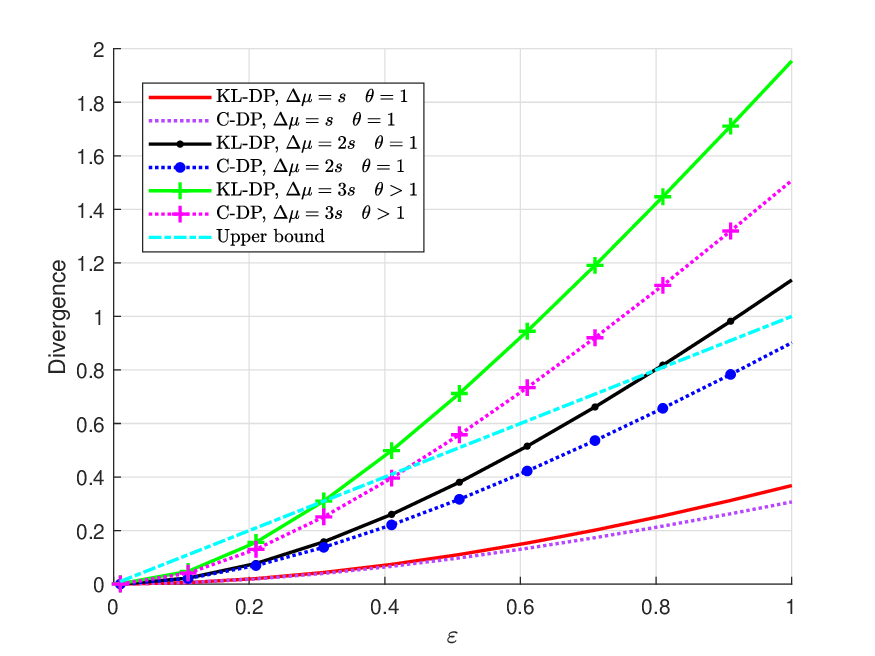}
  \caption{Numerical comparison of KL-DP and Chernoff DP}
  \label{fig:compare_div}
\end{figure}

\section{Adversary's Membership Advantage \label{sec:MIA}}
Membership inference attacks target machine learning (ML) models by attempting to determine whether a given data sample was included in the model's training dataset. These attacks are particularly effective in the presence of overfitting, where the model memorizes the training data rather than learning generalizable patterns, resulting in degraded performance on unseen data. In such cases, the discrepancy between training and non-training samples can be exploited, leading to a high attack success rate.

When models are trained on sensitive data, this vulnerability poses a significant privacy risk to data owners. The adversary's membership advantage is typically defined as the difference between the true positive rate and the false positive rate in deciding whether a specific data point belongs to the training dataset.
In this section, we first review existing upper bounds on the adversary's membership advantage from the literature and then introduce a novel bound based on our proposed approach.

In the original work of \cite{Yeom18}, a data point is defined as a pair $z = (x,y) \in \mathbf{X} \times \mathbf{Y}$, where $x$ denotes the feature vector and $y$ the corresponding response (or label). Let $P_0$ be the underlying data distribution. A training dataset $P_1 \sim P_0^n$ is formed by sampling $n>0$ data points independently and identically distributed (i.i.d.) from $P_0$. The membership inference experiment in \cite{Yeom18} considers an adversary's task of deciding whether a given point $z$ was drawn from the population distribution $P_0$ or from the training dataset $P_1$. After sampling the training data, a bit $b \in {0,1}$ is drawn uniformly. If $b=0$, then $z \sim P_0$; if $b=1$, then $z \sim P_1$. The adversary $\mathcal{A}$ observes the relevant information and outputs a guess $\hat{b} \in {0,1}$.

The adversary's membership advantage is defined as the difference between the following two conditional probabilities
\begin{equation}
\mathrm{Adv}_{\mathrm{MI}}=\Pr[\mathcal{A}=0\mid b=0]-\Pr[\mathcal{A} \mid b=1] \label{advantage} 
\end{equation}
In \cite{BV08}, the prior probabilities associated with the competing hypotheses in the binary hypothesis test are not taken into account in defining the adversary's advantage. Instead, the adversary's advantage is defined as
\begin{equation}
\mathrm{Adv}(P_0,P_1) = 1 - (P_{\mathrm{miss}} + P_{\mathrm{fa}}),
\label{eq:adv}
\end{equation}
where $P_{\mathrm{miss}}$ and $P_{\mathrm{fa}}$ denote the probabilities of miss and false alarm, respectively, as defined in Section~\ref{subsec:hypo_test}.
The expressions in \eqref{advantage} and \eqref{eq:adv} are equivalent, as both quantify the difference between the true positive rate (i.e., the probability of correct detection) and the false positive rate (i.e., the false alarm probability). In particular, since
\begin{equation}
\Pr[\mathcal{A}=0|b=0]=1-\Pr[\mathcal{A}=1|b=0].
\end{equation} the two formulations coincide. In \cite{Shokri,Li13}, the training dataset $D_T$ is sampled from $P_0 \setminus P_1$, in contrast to \cite{Yeom18}, where the queried point $z$ is sampled directly from $P_0$. This distinction may provide the attacker with additional discriminatory information, potentially facilitating the identification of data points that are more likely to have been included in the training dataset.

The following theorem characterizes the optimal advantage achievable when distinguishing between two probability distributions.
\begin{theorem}[Best Advantage \cite{BV08,Crypta04}]
For two probability distributions $P_0$ and $P_1$ of finite support, the best advantage for choosing between $P_0$ and $P_1$ based on $q$ samples, denoted $BestAdv_q(P_0,P_1)$ is given by
\begin{equation}\label{best_adv}
\mathrm{BestAdv}_q(P_0,P_1)=1-2^{-q C(P_0,P_1)}
\end{equation}
\end{theorem} Here, $\mathrm{BestAdv}_q(P_0,P_1)$ represents the maximum possible difference between the power of the test (i.e., the probability of correct detection) and the type I error probability (false alarm rate), optimized over all decision rules.
In \cite{BV08}, the adversary's advantage is defined as in \eqref{eq:adv}, under the assumption that prior weights are not essential in the case of simple hypothesis testing. In this setting, each hypothesis uniquely determines the underlying distribution, unlike in composite hypothesis testing, where prior weights play a meaningful role.

\subsection{Bounding the Advantage\label{subsec:MIAbound}}

In this section, we first review existing upper bounds on the membership inference advantage from the literature and subsequently introduce a novel bound based on Chernoff differential privacy.
\begin{theorem}[\cite{Yeom18}]
Let $\mathcal{M}$ be an $\varepsilon$-differentially private learning algorithm. Then the adversary's membership advantage, $\mathrm{Adv}{\mathrm{MI}}$, is upper bounded as
\begin{equation}\label{Adv_upper1}
\mathrm{Adv}_{\mathrm{MI}}\leq \exp(\varepsilon)-1.
\end{equation}
\end{theorem}
The upper bound on the membership advantage of the adversary has been tightened in \cite{Erl20} down to 
\begin{equation} \label{Adv_upper2}
\mathrm{Adv}_{\mathrm{MI}}\leq 1-\exp(-\varepsilon)+\delta \exp(-\varepsilon)
\end{equation} for an $(\varepsilon,\delta)\textrm{-}$differentially private mechanism. 

The improvement of (\ref{Adv_upper2}) is achieved by combining (\ref{Adv_upper1}) with the proposition of Hall \textit{et al.} in \cite{Hall13} to upper bound the true positive rate as a function of the false positive rate and the privacy budget as
\begin{equation}\label{hall_power}
1-\beta \leq \exp(\varepsilon)\alpha+\delta
\end{equation} Namely, combining (\ref{Adv_upper1}) with (\ref{hall_power}) yields (\ref{Adv_upper2}).
We propose a novel upper bound on the adversary's membership advantage based on Chernoff differential privacy, and show that it coincides with (\ref{Adv_upper2}) in the special case $\delta = 0$.
Setting $P_0=P_{Y|X^n}$ and $P_1=P_{Y|\tilde{X}^n}$ allows us to upper bound the best advantage as a function of the privacy parameter. The resulting bound is stated in the following theorem.
\begin{theorem}[Main result II]\label{theo:MIA_cher}
The adversary's membership inference advantage against an $\varepsilon$-differentially private learning algorithm $\mathcal{M}$ satisfies
\begin{equation} \label{eq:cher_MIA}
\mathrm{Adv}_{\mathrm{MI}} \leq 1-\exp(-\varepsilon)
\end{equation} 
\end{theorem} 
\begin{proof}
Chernoff DP is defined using a natural logarithm. Plugging (\ref{def:C-DP}) into (\ref{best_adv}), in addition to changing the base of the logarithm yields the following original upper bound on the membership advantage
\begin{align} \label{eq:MIA_Chernoff}
\mathrm{Adv}_{\mathrm{MI}} &\leq 1-2^{-\varepsilon \log_2 e} \\
&= 1-\exp(-\varepsilon)
\end{align} for $q=1$ which is the bound given by (\ref{eq:cher_MIA}). 
\end{proof}

The tightest bound currently available in the literature is due to \cite{BayesSec}, and is given by
\begin{equation}
\mathrm{Adv}_{\mathrm{MI}}\leq \frac{\exp{\varepsilon}-1}{\exp{\varepsilon}+1}. \label{eq:bound_bayes}
\end{equation} The bound in (\ref{eq:bound_bayes}) corresponds to the alternative in \cite{Humphries23} for $\delta=0$. 
\subsection{Numerical Evaluation} Numerical comparison of the upper bounds introduced in Section \ref{subsec:MIAbound} are presented in Figures \ref{fig:MIA1}-\ref{fig:MIA2}. The overall privacy budget is taken as $n\varepsilon$ in order to illustrate the impact of composition on the attack performance.
In the figure legends, the original bound \eqref{Adv_upper1} from \cite{Yeom18} is labeled Yeom et al., while the bound in \eqref{Adv_upper2} is labeled Erlingsson for both $\delta=0$ and $\delta>0$. (\ref{eq:cher_MIA}) and (\ref{eq:bound_bayes}) are respectively represented by Chernoff and Bayes security in the legend. Accordingly, (\ref{Adv_upper1}) of \cite{Yeom18} is the loosest one in both scenarios. For the case of $\delta=0$, Chernoff and Erlingsson bounds coincide, yet remain outperformed by the Bayes security bound in (\ref{eq:bound_bayes}) in the high privacy regime for $\varepsilon \leq 4.5 $. For larger values of $\varepsilon$, this gap gradually closes. This effect is even more pronounced in Figure~\ref{fig:MIA2}, where the gap closes at smaller values of $\varepsilon$ (around $\varepsilon \geq 1.5$), reflecting the impact of composition under tighter privacy constraints.
\begin{figure}[ht] \label{figs:MIA12}
			\centering
			\includegraphics[width=0.75\linewidth]{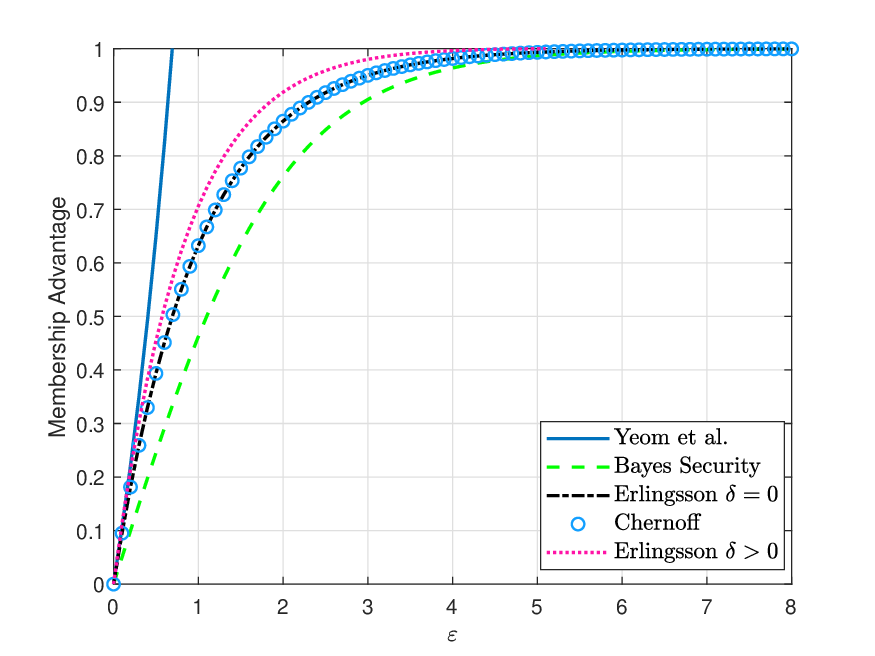}
			\caption{$n=1$}
			\label{fig:MIA1}
	\end{figure}
	\begin{figure}
			\centering
			\includegraphics[width=0.75\linewidth]{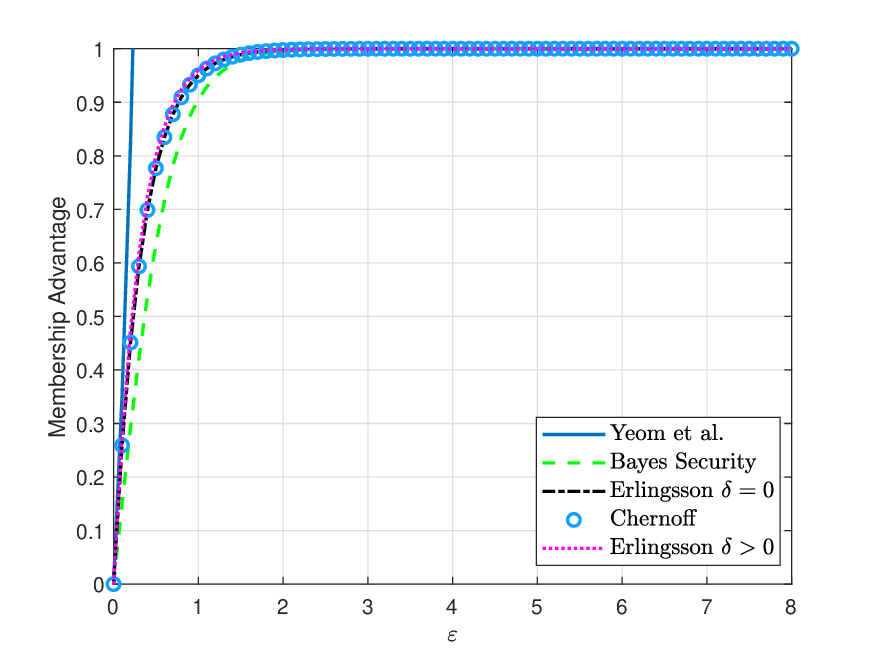}
			\caption{$n=3$}
			\label{fig:MIA2}
\end{figure}

\section{Conclusion \label{sec:conc}}

We studied Chernoff information-namely, the optimal average error exponent for binary classification-through its representation in terms of the Radon-Nikodym derivative, and analyzed its relationship with $\varepsilon$-differential privacy (DP) and KL-DP. We showed that $\varepsilon$-DP implies Chernoff DP, subject to a condition linking the privacy parameter $\varepsilon$ and the Chernoff parameter $\alpha$. This relationship introduces a corresponding constraint on the admissible privacy budget.
The optimal value of $\alpha$, which determines the best average error exponent for a given prior and privacy level $\varepsilon$, could not be obtained in closed form. Instead, we first derived an upper bound on the Radon-Nikodym derivative-based expression of the Chernoff information and then optimized this bound. Developing tighter bounding techniques to enable sharper characterizations remains an important direction for future work. Furthermore, we proved that Chernoff DP implies KL-DP, showing that Chernoff DP is sandwiched between $\varepsilon$-DP and KL-DP in the hierarchy of privacy metrics, ordered from strongest to weakest.

We also presented numerical comparisons between the Kullback-Leibler divergence and Chernoff information in the context of classification under Laplace mechanisms, as functions of the privacy budget, global sensitivity, and the magnitude of the adversarial perturbation. Notably, even when the magnitude of the attack is tripled relative to the query sensitivity, Chernoff information continues to satisfy the derived upper bound in the strong privacy regime, whereas the Kullback-Leibler divergence may violate it.

Finally, we proposed a novel theoretical upper bound on the adversary's membership inference advantage based on Chernoff DP and compared it numerically with existing bounds from the literature under various scenarios. The Chernoff-based bound converges the tightest known bound in the weaker privacy regime (i.e., for larger values of $\varepsilon$, approximately $\varepsilon \geq 4.5$). Future work will investigate additional membership inference settings that may benefit from the rapid convergence of Chernoff information with respect to the number of samples.

\bibliography{paper_adv_dp}

\begin{thebibliography}{10}
\providecommand{\url}[1]{#1}
\csname url@samestyle\endcsname
\providecommand{\newblock}{\relax}
\providecommand{\bibinfo}[2]{#2}
\providecommand{\BIBentrySTDinterwordspacing}{\spaceskip=0pt\relax}
\providecommand{\BIBentryALTinterwordstretchfactor}{4}
\providecommand{\BIBentryALTinterwordspacing}{\spaceskip=\fontdimen2\font plus
\BIBentryALTinterwordstretchfactor\fontdimen3\font minus
  \fontdimen4\font\relax}
\providecommand{\BIBforeignlanguage}[2]{{%
\expandafter\ifx\csname l@#1\endcsname\relax
\typeout{** WARNING: IEEEtran.bst: No hyphenation pattern has been}%
\typeout{** loaded for the language `#1'. Using the pattern for}%
\typeout{** the default language instead.}%
\else
\language=\csname l@#1\endcsname
\fi
#2}}
\providecommand{\BIBdecl}{\relax}
\BIBdecl

\bibitem{AdversarialML}
A.~Joseph, B.~Nelson, B.~Rubinstein, and J.~Tygar, \emph{Adversarial Machine
  Learning}.\hskip 1em plus 0.5em minus 0.4em\relax Cambridge: Cambridge
  University Press, 2018.

\bibitem{Shokri}
\BIBentryALTinterwordspacing
R.~Shokri, M.~Stronati, C.~Song, and V.~Shmatikov, ``Membership inference
  attacks against machine learning models,'' in \emph{2017 IEEE Symposium on
  Security and Privacy (SP)}.\hskip 1em plus 0.5em minus 0.4em\relax Los
  Alamitos, CA, USA: IEEE Computer Society, May 2017, pp. 3--18. [Online].
  Available: \url{https://doi.ieeecomputersociety.org/10.1109/SP.2017.41}
\BIBentrySTDinterwordspacing

\bibitem{Li13}
\BIBentryALTinterwordspacing
N.~Li, W.~Qardaji, D.~Su, Y.~Wu, and W.~Yang, ``Membership privacy: a unifying
  framework for privacy definitions,'' in \emph{Proceedings of the 2013 ACM
  SIGSAC Conference on Computer \& Communications Security}, ser. CCS
  '13.\hskip 1em plus 0.5em minus 0.4em\relax New York, NY, USA: Association
  for Computing Machinery, 2013, p. 889–900. [Online]. Available:
  \url{https://doi.org/10.1145/2508859.2516686}
\BIBentrySTDinterwordspacing

\bibitem{Humphries23}
\BIBentryALTinterwordspacing
T.~Humphries, S.~Oya, L.~Tulloch, M.~Rafuse, I.~Goldberg, U.~Hengartner, and
  F.~Kerschbaum, ``{ Investigating Membership Inference Attacks under Data
  Dependencies },'' in \emph{2023 IEEE 36th Computer Security Foundations
  Symposium (CSF)}.\hskip 1em plus 0.5em minus 0.4em\relax Los Alamitos, CA,
  USA: IEEE Computer Society, Jul. 2023, pp. 473--488. [Online]. Available:
  \url{https://doi.ieeecomputersociety.org/10.1109/CSF57540.2023.00013}
\BIBentrySTDinterwordspacing

\bibitem{MIA_survey}
\BIBentryALTinterwordspacing
H.~Hu, Z.~Salcic, L.~Sun, G.~Dobbie, P.~S. Yu, and X.~Zhang, ``Membership
  inference attacks on machine learning: A survey,'' \emph{ACM Comput. Surv.},
  vol.~54, no. 11s, Sep. 2022. [Online]. Available:
  \url{https://doi.org/10.1145/3523273}
\BIBentrySTDinterwordspacing

\bibitem{DMNA06}
C.~{D}work, F.~{M}cSherry, K.~{N}issim, and A.~{S}mith, ``Calibrating {N}oise
  to {S}ensitivity in {P}rivate {D}ata {A}nalysis,'' in \emph{Theory of
  Cryptography Conference}.\hskip 1em plus 0.5em minus 0.4em\relax
  International Association for Cryptologic Research, 2006, pp. 265--284.

\bibitem{CY16}
P.~{C}uff and L.~{Y}u, ``{D}ifferential {P}rivacy as a {M}utual {I}nformation
  {C}onstraint,'' in \emph{{CCS} 2016, {V}ienna, {A}ustria}.\hskip 1em plus
  0.5em minus 0.4em\relax New York, NY, United States: Association for
  Computing Machinery, Oct. 2016, pp. 43--54.

\bibitem{Renyidp}
I.~Mironov, ``Renyi differential privacy,'' 02 2017.

\bibitem{WYZ16}
W.~{W}ang, L.~{Y}ing, and J.~{Z}hang, ``On the relation between
  identifiability, differential privacy and mutual information privacy,''
  \emph{IEEE Transactions on Information Theory}, vol.~62, pp. 5018--5029, Sep.
  2016.

\bibitem{BK11}
G.~{B}arthe and B.~{K}\"opf, ``Information-theoretic bounds for differentially
  private mechanisms,'' in \emph{Computer Security Foundations
  Symposium}.\hskip 1em plus 0.5em minus 0.4em\relax New York, NY, USA: IEEE,
  2011, pp. 191--204.

\bibitem{M12}
D.~{M}ir, ``Information theoretic foundations of differential privacy,'' in
  \emph{International Symposium of Foundations on Practice of Security}.\hskip
  1em plus 0.5em minus 0.4em\relax Berlin, Heidelberg: Springer, Oct. 2012, pp.
  374--381.

\bibitem{UO2021}
A.~\"Unsal and M.~\"Onen, ``A statistical threshold for adversarial
  classification in laplace mechanisms,'' in \emph{2021 IEEE Information Theory
  Workshop (ITW)}, 2021, pp. 1--6.

\bibitem{UO22}
\BIBentryALTinterwordspacing
------, ``Calibrating the attack to sensitivity in differentially private
  mechanisms,'' \emph{Journal of Cybersecurity and Privacy}, vol.~2, no.~4, pp.
  830--852, 2022. [Online]. Available:
  \url{https://www.mdpi.com/2624-800X/2/4/42}
\BIBentrySTDinterwordspacing

\bibitem{AA11}
M.~S. Alvim, M.~E. Andr{\'e}s, K.~Chatzikokolakis, P.~Degano, and
  C.~Palamidessi, ``Differential privacy: On the trade-off between utility and
  information leakage,'' in \emph{Formal Aspects of Security and Trust}.\hskip
  1em plus 0.5em minus 0.4em\relax Berlin, Heidelberg: Springer Berlin
  Heidelberg, 2012, pp. 39--54.

\bibitem{Unsal_CS}
\BIBentryALTinterwordspacing
A.~\"{U}nsal and M.~\"{O}nen, ``Information-theoretic approaches to
  differential privacy,'' \emph{ACM Comput. Surv.}, vol.~56, no.~3, Oct. 2023.
  [Online]. Available: \url{https://doi.org/10.1145/3604904}
\BIBentrySTDinterwordspacing

\bibitem{JG1}
D.~Johnson, C.~Gruner, K.~Baggerly, and C.~Seshagiri, ``Information-theoretic
  analysis of neural coding,'' \emph{Journal of Computational Neuroscience},
  no.~10, pp. 47--69, 2001.

\bibitem{nov24}
N.~Novello and A.~M. Tonello, ``$f$-divergence based classification: Beyond the
  use of cross-entropy,'' 2024.

\bibitem{Duchi16}
\BIBentryALTinterwordspacing
J.~C. Duchi, K.~Khosravi, and F.~Ruan, ``Information measures, experiments,
  multi-category hypothesis tests, and surrogate losses,'' \emph{ArXiv}, vol.
  abs/1603.00126, 2016. [Online]. Available:
  \url{https://api.semanticscholar.org/CorpusID:13582051}
\BIBentrySTDinterwordspacing

\bibitem{Moreno}
P.~J. Moreno, P.~P. Ho, and N.~Vasconcelos, ``A kullback-leibler divergence
  based kernel for svm classification in multimedia applications,'' in
  \emph{Proceedings of the 16th International Conference on Neural Information
  Processing Systems}, ser. NIPS'03.\hskip 1em plus 0.5em minus 0.4em\relax
  Cambridge, MA, USA: MIT Press, 2003, p. 1385–1392.

\bibitem{Hero2002}
\BIBentryALTinterwordspacing
A.~O. Hero, B.~Ma, O.~J.~J. Michel, and J.~D. Gorman, ``Alpha-divergence for
  classification, indexing and retrieval (revised 2),'' 2002. [Online].
  Available: \url{https://api.semanticscholar.org/CorpusID:12727488}
\BIBentrySTDinterwordspacing

\bibitem{chernoff_fair}
\BIBentryALTinterwordspacing
A.~Nichani, H.~Hsu, Chun-Fu, Chen, and H.~Jeong, ``Does privacy always harm
  fairness? data-dependent trade-offs via chernoff information neural
  estimation,'' 2026. [Online]. Available:
  \url{https://arxiv.org/abs/2601.13698}
\BIBentrySTDinterwordspacing

\bibitem{chernoff_fair2}
A.~Nichani, H.~Hsu, and H.~Jeong, ``Can we catch the two birds of fairness and
  privacy?'' in \emph{International Conference on Learning Representations
  (ICLR), 2025 Advances in Financial AI Workshop}, 2025.

\bibitem{Zari22}
O.~Zari, J.~Parra-Arnau, A.~{\"U}nsal, T.~Strufe, and M.~{\"O}nen, ``Membership
  inference attack against principal component analysis,'' in \emph{Privacy in
  Statistical Databases}, J.~Domingo-Ferrer and M.~Laurent, Eds.\hskip 1em plus
  0.5em minus 0.4em\relax Cham: Springer International Publishing, 2022, pp.
  269--282.

\bibitem{Yeom18}
\BIBentryALTinterwordspacing
S.~Yeom, I.~Giacomelli, M.~Fredrikson, and S.~Jha, ``{ Privacy Risk in Machine
  Learning: Analyzing the Connection to Overfitting },'' in \emph{2018 IEEE
  31st Computer Security Foundations Symposium (CSF)}.\hskip 1em plus 0.5em
  minus 0.4em\relax Los Alamitos, CA, USA: IEEE Computer Society, Jul. 2018,
  pp. 268--282. [Online]. Available:
  \url{https://doi.ieeecomputersociety.org/10.1109/CSF.2018.00027}
\BIBentrySTDinterwordspacing

\bibitem{BayesSec}
\BIBentryALTinterwordspacing
K.~Chatzikokolakis, G.~Cherubin, C.~Palamidessi, and C.~Troncoso, ``Bayes
  security measure,'' Dec. 2020, working paper or preprint. [Online].
  Available: \url{https://inria.hal.science/hal-03091416}
\BIBentrySTDinterwordspacing

\bibitem{Erl20}
\BIBentryALTinterwordspacing
\'{U}lfar Erlingsson, I.~Mironov, A.~Raghunathan, and S.~Song, ``That which we
  call private,'' 2020. [Online]. Available:
  \url{https://arxiv.org/abs/1908.03566}
\BIBentrySTDinterwordspacing

\bibitem{Hall13}
R.~Hall, A.~Rinaldo, and L.~Wasserman, ``Differential privacy for functions and
  functional data,'' \emph{J. Mach. Learn. Res.}, vol.~14, no.~1, p. 703–727,
  Feb. 2013.

\bibitem{Prob_Measure}
P.~Billingsley, \emph{Probability and Measure}.\hskip 1em plus 0.5em minus
  0.4em\relax New York: Wiley, 1995.

\bibitem{Nikodym}
\BIBentryALTinterwordspacing
O.~Nikodym, ``Sur une g\'en\'eralisation des int\'egrales de m. j. radon,''
  \emph{Fundamenta Mathematicae}, vol.~15, no.~1, pp. 131--179, 1930. [Online].
  Available: \url{http://eudml.org/doc/212339}
\BIBentrySTDinterwordspacing

\bibitem{DR05}
C.~{D}work and A.~{R}oth, ``The {A}lgorithmic {F}oundations of {D}ifferential
  {P}rivacy,'' \emph{Foundations and Trends in Theoretical Computer Science
  2014}, vol.~9, pp. 211--407, 2014.

\bibitem{Kullback}
S.~Kullback, \emph{Information Theory and Statistics}.\hskip 1em plus 0.5em
  minus 0.4em\relax New York: Wiley, 1959.

\bibitem{KL}
S.~Kullback and R.~Leibler, ``On information and sufficiency,'' \emph{Annals of
  Mathematical Statistics}, vol.~22, 1951.

\bibitem{Chernoff}
H.~Chernoff, ``A measure of asymptotic efficiency for tests of a hypothesis
  based on the sum of observations,'' \emph{Annals of Mathematical Statistics},
  vol.~23, pp. 493--507, 1952.

\bibitem{r14}
T.~Cover and J.~A. Thomas, \emph{Elements of Information Theory}.\hskip 1em
  plus 0.5em minus 0.4em\relax Wiley Series in Telecommunications, 1991.

\bibitem{DKMMN}
C.~Dwork, K.~Kenthapadi, F.~McSherry, I.~Mironov, and M.~Naor, ``Our data,
  ourselves: Privacy via distributed noise generation,'' in \emph{Advances in
  Cryptology - EUROCRYPT 2006}, S.~Vaudenay, Ed.\hskip 1em plus 0.5em minus
  0.4em\relax Berlin, Heidelberg: Springer Berlin Heidelberg, 2006, pp.
  486--503.

\bibitem{KOV15}
P.~{K}airouz, S.~{O}h, and P.~{V}iswanath, ``The composition theorem for
  differential privacy,'' in \emph{32nd International Conference on Machine
  Learning}.\hskip 1em plus 0.5em minus 0.4em\relax JMLR, Inc. and Microtome
  Publishing (United States), 2015, pp. 4037--4049.

\bibitem{JS03}
D.~Johnson and S.~Sinanovic, ``Symmetrizing the kullback-leibler distance,'' 02
  2003.

\bibitem{BV08}
T.~Baign{\`e}res and S.~Vaudenay, ``The complexity of distinguishing
  distributions (invited talk),'' in \emph{Information Theoretic Security},
  R.~Safavi-Naini, Ed.\hskip 1em plus 0.5em minus 0.4em\relax Berlin,
  Heidelberg: Springer Berlin Heidelberg, 2008, pp. 210--222.

\bibitem{Crypta04}
T.~Baign{\`e}res, P.~Junod, and S.~Vaudenay, ``How far can we go beyond linear
  cryptanalysis?'' in \emph{Advances in Cryptology - ASIACRYPT 2004}, P.~J.
  Lee, Ed.\hskip 1em plus 0.5em minus 0.4em\relax Berlin, Heidelberg: Springer
  Berlin Heidelberg, 2004, pp. 432--450.

\end{thebibliography}
\bibliographystyle{IEEEtran}

\end{document}